\documentclass[1p]{elsarticle}
\journal{}
\usepackage{fix-cm}
\usepackage{mathtools}
\usepackage{amsmath,amsthm}
\usepackage{graphicx}
\usepackage{latexsym,enumerate,amssymb,amsbsy}
\usepackage{graphics,color}
\usepackage{verbatim}
\usepackage{url}
\usepackage{dblfloatfix}
\usepackage{mathptmx}
\usepackage{tikz-network}
\usepackage{booktabs}
\usepackage{xfrac}
\usepackage{amsfonts}
\usepackage{stmaryrd}
\usepackage{algorithm}
\usepackage[noend]{algpseudocode}
\usepackage{url}
\usepackage{colortbl}
\usepackage{adjustbox}
\usepackage{array,longtable}
\newtheorem{proposition}{Proposition}
\newtheorem{theorem}[proposition]{Theorem}

\newtheorem{example}{Example}

\newtheorem{definition}{Definition}

\def\F{{\mathbb{F}}}

\def\Z{{\mathbb{Z}}}

\def\0{{\mathbf{0}}}
\def\1{{\mathbf{1}}}

\DeclarePairedDelimiter{\floor}{\lfloor}{\rfloor}

\newcommand\dsb[1]{\llbracket #1 \rrbracket}

\begin{document}
	
\begin{frontmatter}
\title{Two New Zero-Dimensional Qubit Codes from Bordered Metacirculant Construction}
	
\author[tamu]{Padmapani Seneviratne}
\ead{Padmapani.Seneviratne@tamuc.edu}
	
\cortext[cor1]{Corresponding author}	
\author[ntu]{Martianus Frederic Ezerman\corref{cor1}}
\ead{fredezerman@ntu.edu.sg}
	
\address[tamu]{Department of Mathematics, Texas A\&M University-Commerce,\\
2600 South Neal Street, Commerce TX 75428, U.S.A.}
	
\address[ntu]{School of Physical and Mathematical Sciences, Nanyang Technological University,\\
21 Nanyang Link, Singapore 637371.}
	
\begin{abstract}
We construct qubit stabilizer codes with parameters $\llbracket 81, 0, 20 \rrbracket_2$ and $\llbracket 94, 0, 22 \rrbracket_2$ for the first time. We use symplectic self-dual additive codes over $\F_4$ built by modifying the adjacency matrices of suitable metacirculant graphs found by a randomized search procedure.
\end{abstract}
	
\end{frontmatter}


\section{Introduction}\label{sec:intro}
Let $\omega$ be a root of $x^2+x+1 \in \F_2[x]$. The finite field of $4$ element is $\mathbb{F}_4:=\{0,1,\omega,\omega^2=\overline{\omega}\}$. A code $C$ of length $n$ over $\mathbb{F}_4$ is {\it additive} if the sum of any two, not necessarily distinct, codewords $\mathbf{c},\mathbf{w} \in C$ is again a codeword in $C$. The {\it weight} of a vector in $\mathbb{F}_4^n$ is the number of its nonzero entries. The {\it distance} of vectors $\mathbf{a}, \mathbf{b} \in \mathbb{F}_4^n$ is the number of positions where their respective entries differ. The {\it minimum distance} of an additive code $C$ is the smallest weight $d$ among the weights of its nonzero codewords. The code is Type II if each of its codewords has an even weight. Otherwise, $C$ is Type I. If $|C|=2^k$, then it has the parameters $(n,2^k,d)_4$. 

The {\it trace Hermitian inner product} of $\mathbf{a},\mathbf{b} \in \F_4^n$ is 
\begin{equation}
\mathbf{a} * \mathbf{b} = \sum_{j=1}^n \left( a_j \, b_j^2 + a_j^2 \, b_j\right).
\end{equation}
The code $C^* = \{ \mathbf{w} \in \F_4^n : \mathbf{w} * \mathbf{c} =0 \mbox{ for all } \mathbf{c} \in C\}$ is the {\it symplectic dual} of $C$ and $C$ is {\it self-dual} if $C = C^*$. A self-dual additive code $C$ over $\mathbb{F}_4$ directly leads to a quantum bit (qubit) {\it stabilizer code} $Q:=Q(C)$ with parameters $\dsb{n,0,d}_2$. This link is so well-established that we simply refer the readers to the exposition of Calderbank, Rains, Shor, and Sloane in~\cite{Calderbank1998} for the details. The zero dimensional qubit code $Q$ represents a \emph{highly entangled single quantum state}, useful when high entanglement and error-correction rate are simultaneously required. 

Let $G$ be a simple undirected graph with $n$ vertices and {\it adjacency matrix} $\Gamma:=\Gamma(G)$. A {\it graph code} on $G$ is an $(n,2^n,d)_4$ additive code $C:=C(G)$ whose codewords are the $\mathbb{F}_2$-linear combinations of the rows of $\Gamma + \omega \, I_n$. Any graph code is clearly self-dual. 
Various versions and proofs of the following main result, first discussed by Schlingemann in~\cite{Schlingemann2002}, are nicely presented by Danielsen and Parker in \cite[Section~3]{Danielsen2006}. Every graph represents a self-dual additive code over $\mathbb{F}_4$ and every self-dual additive code over $\mathbb{F}_4$ can be represented by a graph. Thus, to construct $\dsb{n,0,d}_2$ stabilizer codes, it suffices to find graph codes with minimum distance $d$. Two graph codes are {\it equivalent} if their corresponding graphs are isomorphic \cite[Thm.~6]{Danielsen2006}.

A complete classification of \emph{all} self-dual additive codes over $\mathbb{F}_4$ for $n \leq 12$ was done in~\cite{Danielsen2006}. Follow-up works, covering $n \leq 50$, were done by Gulliver and Kim in~\cite{GK2004}, by Varbanov in~\cite{Varbanov2008}, by Grassl and Harada in~\cite{Grassl2017}, and by Saito in~\cite{Saito2019}. Their efforts were limited mostly to graphs whose adjacency matrices are {\it circulant} or {\it bordered circulant}. Excellent $d$ has been observed to come from strongly regular graphs with large symmetry groups. As $n$ grows, the challenge to search for, much less classify, codes with optimal $d$ escalates since computing the exact minimum distance is hard~\cite{Vardy1997}. 

In \cite{SE2020} we restricted the graphs to the {\it metacirculant} family. 
In the present work, the new $\dsb{81, 0, 20}_2$ and $\dsb{94,0,24}_2$ codes come from {\it bordered metacirculant graphs}. For $n \in \{29,37\}$ we obtain bordered metacirculant graph codes with strictly better minimum distances than the best bordered circulant ones.

\section{Bordered Metacirculant Graphs}\label{sec:borderedmeta}

\begin{definition} (\cite{Alspach1982})\label{def:mc}
Let $m,\ell \in \mathbb{N}$ be fixed and $\alpha\in \Z_{\ell}$ be a unit. Let $S_{0}, S_{1}, \ldots, S_{\floor{m/2}} \subseteq \Z_{\ell}$ satisfy the four properties
\begin{enumerate}
	\item $S_{0} = -S_{0}$,
	\item $0\notin S_{0}$,
	\item $\alpha^{m}S_k = S_k$ for $1\le k\le \lfloor m/2 \rfloor$,
	\item if $m$ is even, then $\alpha^{m/2}S_{m/2} = -S_{m/2}$. 
\end{enumerate}
The {\it metacirculant} graph $G:=G\left(m, \ell, \alpha, S_{0}, S_{1}, \ldots, S_{\floor{m/2}}\right)$ has as its vertex set $V(G):=\Z_{m}\times \Z_{\ell}$ of size $n = m \cdot \ell$. 
Let $V_0, V_1, \ldots, V_{m-1}$, where 
$V_{i} = \{(i, j) : 0 \le j < \ell \}$, be a partition of $V(G)$. Let $1 \le k \le \floor{m/2}$. Vertices $(i, j)$ and $(i+k, h)$ are adjacent if and only if $(h-j) \in \alpha^{i} \, S_k$.
\end{definition}

The adjacency matrix $\Gamma:=\Gamma(G)$ of a metacirculant graph $G$ is a \textit{metacirculant matrix}. A \textit{bordered metacirculant matrix}  \begin{equation*}
\overline{\Gamma} = 
\begin{bmatrix}
0 & 1 & \cdots & 1\\
1 &   &        &  \\
\vdots &  & \Gamma &   \\
1 &    &      &   \\
\end{bmatrix},
\end{equation*}
is the adjacency matrix of the graph $\overline{G}$, obtained from $G$ by adding a new vertex $v_{\infty}$ and joining $v_{\infty}$ to each vertex of $G$.

\begin{figure}[h!]
\centering
\begin{tikzpicture}[multilayer=3d]
		\SetLayerDistance{-1.0}
		\Plane[x=-.4,y=-.4,width=2.6,height=1.8,color=gray,layer=2,NoBorder]
		\Plane[x=-.4,y=-.4,width=2.6,height=1.8,NoBorder]
		\Vertex[IdAsLabel,layer=1]{1}
		\Vertex[x=0.4,y=1.0,IdAsLabel,layer=1]{2}
		\Vertex[x=1.7,y=0.5,IdAsLabel,layer=1]{3}
		\Vertex[IdAsLabel,color=gray,layer=2]{4}
		\Vertex[x=0.4,y=1.0,IdAsLabel,style=dashed,color=gray,layer=2]{5}
		\Vertex[x=1.7,y=0.5,IdAsLabel,color=gray,layer=2]{6}
		\Edge(1)(2)
		\Edge(1)(3)
		\Edge(2)(3)
		\Edge[style=dashed](4)(5)
		\Edge[style=dashed](6)(5)
		\Edge(4)(6)
		\Edge[color=blue](1)(4)
		\Edge[style=dashed,color=blue](2)(5)
		\Edge[color=blue](3)(6)
\end{tikzpicture}
\caption{$G_{6}:=G(2, 3, 1, \{1, 2\}, \{0\})$ that yields the {\tt hexacode}.}
\label{fig:6}
\end{figure}
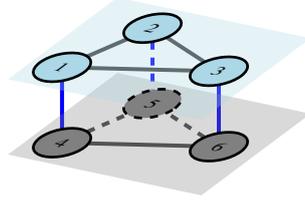
\begin{example}
The unique $\dsb{6,0,4}_2$ code is from the $(6,2^6,4)_4$ {\tt hexacode}, which can be derived from the metacirculant graph $G_{6}:=G(2, 3, 1, \{1, 2\}, \{0\})$. Figure~\ref{fig:6} labels vertices in $V_{0}=\{(0,0),(0,1),(0,2)\}$ and $V_{1}=\{(1,0),(1,1),(1,2)\}$ as $1,2,\ldots,6$. The adjacency matrix $\overline{\Gamma}_6$ of $\overline{G}_6$ and the generator matrix $\overline{\Gamma}_{6} + \omega \, I_{7}$ of the self-dual $(7,2^7,3)_4$ code are, respectively,
\[
\begin{bmatrix}
0 & 1 & 1 & 1 & 1 & 1 & 1 \\
1 & 0 & 1 & 1 & 1 & 0 & 0 \\
1 & 1 & 0 & 1 & 0 & 1 & 0 \\
1 & 1 & 1 & 0 & 0 & 0 & 1 \\
1 & 1 & 0 & 0 & 0 & 1 & 1 \\
1 & 0 & 1 & 0 & 1 & 0 & 1 \\
1 & 0 & 0 & 1 & 1 & 1 & 0 
\end{bmatrix}
\mbox{ and } 
\begin{bmatrix}
\omega & 1 & 1 & 1 & 1 & 1 & 1 \\
1 & \omega & 1 & 1 & 1 & 0 & 0 \\
1 & 1 & \omega & 1 & 0 & 1 & 0 \\
1 & 1 & 1 & \omega & 0 & 0 & 1 \\
1 & 1 & 0 & 0 & \omega & 1 & 1 \\
1 & 0 & 1 & 0 & 1 & \omega & 1 \\
1 & 0 & 0 & 1 & 1 & 1 & \omega 
\end{bmatrix}.
\]
\end{example}

The highest $d$ of an additive self-dual $(29,2^{29},d)_4$ codes from the bordered circulant graph construction is $d=9$, as shown in~\cite{Saito2019}. Using the bordered metacirculant graph construction, we obtain a self-dual $(29,2^{29},10)_4$ code $C_{29}$. The strict improvement in $d$ in Proposition~\ref{prop29} does not carry over to the quantum setup. The code $C_{29}$ gives a $\dsb{29,0,10}_2$ stabilizer code. A better $\dsb{29,0,11}_2$ code from a circulant graph code was recorded in~\cite[Table II]{GK2004}. Its corresponding $(29,2^{29},11)_4$ circulant code is unique, up to equivalence~\cite{Varbanov2008}.

\begin{proposition}\label{prop29}
The bordered metacirculant graph 
$\overline{G}_{28}$ with 
\[
G_{28}:=G(2,14,13, \{5, 6, 8, 9\}, \{0, 1, 3, 6, 7, 9, 11\})
\] 
in Table~\ref{graph:28} generates a self-dual $(29,2^{29},10)_4$ code $C_{29}$.
\end{proposition}

\begin{theorem}\label{thm:even}
Let $C$ be an additive symplectic self-dual code over $\F_4$ generated by 
a bordered metacirculant $\overline {G}\left(m,\ell,\alpha, S_{0}, S_{1}, \ldots, S_{\floor{m/2}}\right)$. Let $\Delta_S :=|S_0| + 1$ if $m$ is odd and 
$\Delta_S :=|S_0| + |S_{\floor{m/2}}| + 1$ if $m$ is even. 
Then $C$ is Type II if and only if both $\Delta_S$ and $n = m \cdot \ell$ are odd.
\end{theorem}

\begin{proof}
An additive graph code $C(G)$ is Type II if and only if all vertices of $G$ are of odd degree~\cite{Danielsen2006}. In the bordered metacirculant 
$\overline{G}\left(m,\ell,\alpha, S_{0}, S_{1}, \ldots, S_{\floor{m/2}}\right)$, $\deg(v_{\infty}) = m \cdot \ell$ and
\begin{align}\label{eq:degree}
\deg(v_i) &=
\begin{cases} 
|S_0| + |S_1| + 1 \mbox{ when } m=2,\\
|S_0| + |S_{\floor{ m/2}}| + 
2 \left(\sum_{r=1}^{\floor{m/2}-1} |S_r|\right) + 1 \mbox{ for even } m \geq 4,\\
|S_0| + 2 \left(\sum_{r=1}^{\floor{m/2}} |S_r|\right) + 1 \mbox{ for odd } m \geq 3.
\end{cases}
\end{align}
Thus follows the desired conclusion.
\end{proof}

\begin{table}[h!]
\caption{The $154$ edges of $G_{28}$ as $(i,\{j \in J\})$. 
For $0 \leq t <14$, vertices $(0,t)$ are $1$ to $14$ and vertices $(1,t)$ are $15$ to $28$.}
\label{graph:28}
\centering
\begin{tabular}{l}
\toprule
$(1,\{2, 4, 8, 11, 13, 14, 16, 17, 19, 20, 24 \}),
(2,\{7, 11, 12, 14, 15, 17, 18, 20, 23, 27 \}),$\\ 
$(3,\{4, 6, 10, 13, 15, 16, 18, 19, 21, 22, 26 \}),
(4,\{9, 13, 14, 16, 17, 19, 20, 22, 25 \}),$\\
$(5,\{6, 8, 12, 15, 17, 18, 20, 21, 23, 24, 28 \}),
(6,\{11, 15, 16, 18, 19, 21, 22, 24, 27 \}),$\\
$(7,\{8, 10, 14, 17, 19, 20, 22, 23, 25, 26 \}),
(8,\{13, 17, 18, 20, 21, 23, 24, 26 \}),$\\
$(9,\{10, 12, 16, 19, 21, 22, 24, 25, 27, 28 \}),
(10,\{15, 19, 20, 22, 23, 25, 26, 28 \}),$\\
$(11,\{12, 14, 18, 21, 23, 24, 26, 27 \})
(12,\{17, 21, 22, 24, 25, 27, 28 \}),$\\
$(13,\{14, 16, 20, 23, 25, 26, 28 \}),
(14,\{19, 23, 24, 26, 27 \}),
(15,\{16, 18, 22, 25, 27, 28 \}),$\\
$(16,\{21, 25, 26, 28 \}),
(17,\{18, 20, 24, 27 \}),
(18,\{23, 27, 28 \}),(19,\{20, 22, 26 \}),$\\
$(20,\{25 \}),(21,\{22, 24, 28 \}),
(22,\{27 \}),(23,\{24, 26 \}),
(25,\{26, 28 \}),(27,\{28 \}).$\\
\bottomrule 
\end{tabular}
\end{table}

\begin{table}
\caption{The $198$ and $342$ respective edges of $G_{36,1}$ and $G_{36,2}$ as $(i,\{j \in J\})$. For $0 \leq t <18$, vertices $(0,t)$ are $1$ to $18$ while $(1,t)$ are $19$ to $36$.}
\label{graph:36}
\centering
\resizebox{\textwidth}{!}{%
\begin{tabular}{l}
\toprule
$(1, \{ 3, 6, 7, 14, 18, 19, 24, 31, 32, 34, 35 \}),
(2, \{ 4, 5, 7, 8, 15, 20, 21, 25, 32, 33, 36 \}),$ \\

$(3, \{ 5, 8, 9, 16, 20, 21, 26, 33, 34, 36 \}),
(4, \{ 6, 7, 9, 10, 17, 22, 23, 27, 34, 35 \}),$ \\

$(5, \{ 7, 10, 11, 18, 22, 23, 28, 35, 36 \}),
(6, \{ 8, 9, 11, 12, 19, 24, 25, 29, 36 \}),$ \\

$(7, \{ 9, 12, 13, 20, 24, 25, 30 \}),
(8, \{ 10, 11, 13, 14, 21, 26, 27, 31 \}),
(9, \{ 11, 14, 15, 22, 26, 27, 32 \}),$\\

$(10, \{ 12, 13, 15, 16, 23, 28, 29, 33 \}),
(11, \{ 13, 16, 17, 24, 28, 29, 34 \}),$\\

$(12, \{ 14, 15, 17, 18, 25, 30, 31, 35 \}),
(13, \{ 15, 18, 19, 26, 30, 31, 36 \}),$\\

$(14, \{ 16, 17, 19, 20, 27, 32, 33 \}),
(15, \{ 17, 20, 21, 28, 32, 33 \}),
(16, \{ 18, 19, 21, 22, 29, 34, 35 \}),$\\

$(17, \{ 19, 22, 23, 30, 34, 35 \}),
(18, \{ 20, 21, 23, 24, 31, 36 \}),
(19, \{ 21, 24, 25, 32, 36 \}),$\\

$(20, \{ 22, 23, 25, 26, 33 \}),
(21, \{ 23, 26, 27, 34 \}),
(22, \{ 24, 25, 27, 28, 35 \}),$\\

$(23, \{ 25, 28, 29, 36 \}),
(24, \{ 26, 27, 29, 30 \}),
(25, \{ 27, 30, 31 \}),
(26, \{ 28, 29, 31, 32 \}),$\\

$(27, \{ 29, 32, 33 \}),
(28, \{ 30, 31, 33, 34 \}),
(29, \{ 31, 34, 35 \}),
(30, \{ 32, 33, 35, 36 \}),$\\

$(31, \{ 33, 36 \}),
(32, \{ 34, 35 \}),
(33, \{ 35 \}),(34, \{ 36 \})$. \\
\midrule

$(1, \{ 2, 3, 4, 5, 6, 7, 8, 11, 12, 18, 20, 24, 26, 27, 28, 31, 32, 33, 35 \}),$ \\

$(2, \{ 4, 6, 7, 8, 11, 12, 13, 15, 19, 21, 27, 28, 31, 32, 33, 34, 35, 36 \}),$ \\

$(3, \{ 4, 5, 6, 7, 8, 9, 10, 13, 14, 20, 22, 26, 28, 29, 30, 33, 34, 35 \}),$ \\

$(4, \{ 6, 8, 9, 10, 13, 14, 15, 17, 21, 23, 29, 30, 33, 34, 35, 36 \}),$ \\

$(5, \{ 6, 7, 8, 9, 10, 11, 12, 15, 16, 22, 24, 28, 30, 31, 32, 35, 36 \}),$ \\

$(6, \{ 8, 10, 11, 12, 15, 16, 17, 19, 23, 25, 31, 32, 35, 36 \}),$ \\

$(7, \{ 8, 9, 10, 11, 12, 13, 14, 17, 18, 24, 26, 30, 32, 33, 34 \}),$\\

$(8, \{ 10, 12, 13, 14, 17, 18, 19, 21, 25, 27, 33, 34 \}),$ \\

$(9, \{ 10, 11, 12, 13, 14, 15, 16, 19, 20, 26, 28, 32, 34, 35, 36 \}),$ \\ 

$(10, \{ 12, 14, 15, 16, 19, 20, 21, 23, 27, 29, 35, 36 \}),$\\

$(11, \{ 12, 13, 14, 15, 16, 17, 18, 21, 22, 28, 30, 34, 36 \}),
(12, \{ 14, 16, 17, 18, 21, 22, 23, 25, 29, 31 \}),$\\

$(13, \{ 14, 15, 16, 17, 18, 19, 20, 23, 24, 30, 32, 36 \}),
(14, \{ 16, 18, 19, 20, 23, 24, 25, 27, 31, 33 \}),$\\

$(15, \{ 16, 17, 18, 19, 20, 21, 22, 25, 26, 32, 34 \}),
(16, \{ 18, 20, 21, 22, 25, 26, 27, 29, 33, 35 \}),$\\

$(17, \{ 18, 19, 20, 21, 22, 23, 24, 27, 28, 34, 36 \}),
(18, \{ 20, 22, 23, 24, 27, 28, 29, 31, 35 \}),$\\

$(19, \{ 20, 21, 22, 23, 24, 25, 26, 29, 30, 36 \}),
(20, \{ 22, 24, 25, 26, 29, 30, 31, 33 \}),$\\

$(21, \{ 22, 23, 24, 25, 26, 27, 28, 31, 32 \}),
(22, \{ 24, 26, 27, 28, 31, 32, 33, 35 \}),$\\

$(23, \{ 24, 25, 26, 27, 28, 29, 30, 33, 34 \}),
(24, \{ 26, 28, 29, 30, 33, 34, 35 \}),$\\

$(25, \{ 26, 27, 28, 29, 30, 31, 32, 35, 36 \}),
(26, \{ 28, 30, 31, 32, 35, 36 \}),$\\

$(27, \{ 28, 29, 30, 31, 32, 33, 34 \}),
(28, \{ 30, 32, 33, 34 \}),
(29, \{ 30, 31, 32, 33, 34, 35, 36 \}),$\\

$(30, \{ 32, 34, 35, 36 \}),(31, \{ 32, 33, 34, 35, 36 \}),
(32, \{ 34, 36 \}),$\\

$(33, \{ 34, 35, 36 \}),(34, \{ 36 \}),(35, \{ 36 \})$.\\
\bottomrule
\end{tabular}
}
\end{table}

There exists a unique, up to equivalence, $(37,2^{37}, 11)_4$ additive self-dual code from bordered circulant construction~\cite{Saito2019}. Using the bordered metacirculant construction, we found two inequivalent codes. The complete edge sets of the corresponding metacirculant graphs are given in Table~\ref{graph:36}.
\begin{proposition}
There are at least two inequivalent self-dual additive $(37, 2^{37}, 11)_4$ codes, built respectively from the graphs 
\begin{align}
\overline{G}_{36,1} &:= \overline{G}(2, 18, 17, \{ 1, 3, 9, 15, 17 \}, \{2,6, 8, 11, 15, 16 \}) \mbox{ and}\\
\overline{G}_{36,2} &:= \overline{G}(2, 18, 17,
\{ 1, 2, 3, 5, 13, 15, 16, 17 \},\{ 0, 1, 2, 3, 5, 8, 9, 11, 12, 13, 15 \}).
\end{align}
The codes have, respectively, $252$ and $270$ codewords of weight $11$. The corresponding $\dsb{37,0,11}_2$ code has the best-known parameters~\cite{Grassl:codetables}.
\end{proposition}

\begin{table}
\caption{The $1640$ edges of $G_{80,1}$ as 
$(i,\{j \in J\}),$ where $i,j$ are indices of the adjacent vertices. 
Vertices are partitioned into $V_0,V_1,\ldots,V_7$, with the $10$ vertices in 
$\{(k,t) : 0 \leq k < 8,\, 0 \leq t < 10\}$ in $V_k$. They are renamed as vertices $z,z+1,\ldots,z+9$ for $0 \leq k < 8$ and $z=10 k + 1$.}
\label{graph:80}
\renewcommand{\arraystretch}{1.05}
\centering
\resizebox{\textwidth}{!}{%
\begin{tabular}{l}
\toprule
$(1, \{ 2, 3, 4, 5, 6, 7, 8, 9, 10, 13, 16, 18, 19, 21, 23, 24, 26, 27, 31, 33, 35, 37, 39, 40, 49, 50, 52, 53, 56, 58, 64, 66, 67, 69, 70, 71, 73, 75, 77, 79, 80 \}),$ \\ 
$(2, \{ 3, 4, 5, 6, 7, 8, 11, 12, 16, 17, 18, 19, 21, 22, 25, 26, 28, 30, 32, 33, 35, 36, 38, 40, 51, 52, 54, 55, 56, 57, 58, 59, 62, 65, 66, 68, 70, 72, 73, 75 \}),$ \\ 
$(3, \{ 4, 5, 6, 7, 8, 9, 10, 11, 13, 15, 17, 20, 21, 23, 24, 26, 28, 34, 35, 36, 38, 39, 49, 50, 51, 53, 55, 57, 60, 61, 65, 66, 68, 69, 71, 74, 75, 76, 79 \}),$ \\ 
$(4, \{ 5, 6, 7, 8, 11, 13, 18, 19, 20, 22, 24, 27, 28, 29, 32, 33, 34, 37, 38, 40, 50, 51, 53, 54, 56, 58, 59, 60, 62, 64, 67, 68, 69, 71, 72, 74, 77, 78 \}),$ \\ 
$(5, \{ 6, 7, 8, 9, 12, 13, 14, 17, 19, 20, 22, 23, 27, 30, 31, 33, 35, 36, 37, 39, 49, 52, 53, 54, 56, 60, 62, 65, 66, 67, 70, 71, 73, 75, 76, 77, 79 \}),$ \\ 

$(6, \{ 7, 8, 12, 15, 16, 17, 18, 21, 22, 23, 26, 28, 29, 30, 32, 34, 36, 37, 39,40, 50, 51, 52, 55, 56, 58, 61, 62, 63, 66, 68, 69, 70, 72, 77, 79 \}),$ \\ 
$(7, \{ 8, 9, 11, 13, 14, 15, 17, 19, 20, 21, 24, 30, 32, 34, 35, 38, 39, 40, 49,51, 53, 54, 55, 57, 61, 64, 65, 67, 69, 70, 72, 75, 78, 79, 80 \}),$ \\ 
$(8, \{ 9, 15, 18, 20, 22, 23, 24, 25, 28, 31, 32, 33, 34, 36, 37, 38, 49, 50, 52, 54, 55, 58, 60, 62, 63, 64, 65, 67, 68, 71, 72, 73, 74, 78 \}),$ \\ 
$(9, \{ 10, 11, 12, 13, 14, 15, 16, 17, 18, 21, 24, 26, 27, 29, 31, 32, 34, 35, 39, 41, 43, 45, 47, 48, 57, 58, 60, 61, 64, 66, 72, 74, 75, 77, 78, 79 \}),$ \\ 
$(10, \{ 11, 12, 13, 14, 15, 16, 19, 20, 24, 25, 26, 27, 29, 30, 33, 34, 36, 38, 40, 41, 43, 44, 46, 48, 59, 60, 62, 63, 64, 65, 66, 67, 70, 73, 74, 76, 78, 80 \}),$ \\

$(11, \{ 12, 13, 14, 15, 16, 17, 18, 19, 21, 23, 25, 28, 29, 31, 32, 34, 36, 42, 43, 44, 46, 47, 57, 58, 59, 61, 63, 65, 68, 69, 73, 74, 76, 77, 79 \}),$ \\ 
$(12, \{ 13, 14, 15, 16, 19, 21, 26, 27, 28, 30, 32, 35, 36, 37, 40, 41, 42, 45, 46, 48, 58, 59, 61, 62, 64, 66, 67, 68, 70, 72, 75, 76, 77, 79, 80 \}),$ \\ 
$(13, \{ 14, 15, 16, 17, 20, 21, 22, 25, 27, 28, 30, 31, 35, 38, 39, 41, 43, 44, 45, 47, 57, 60, 61, 62, 64, 68, 70, 73, 74, 75, 78, 79 \}),$ \\ 
$(14, \{ 15, 16, 20, 23, 24, 25, 26, 29, 30, 31, 34, 36, 37, 38, 40, 42, 44, 45, 47, 48, 58, 59, 60, 63, 64, 66, 69, 70, 71, 74, 76, 77, 78, 80 \}),$ \\ 
$(15, \{ 16, 17, 19, 21, 22, 23, 25, 27, 28, 29, 32, 38, 40, 42, 43, 46, 47, 48, 57, 59, 61, 62, 63, 65, 69, 72, 73, 75, 77, 78, 80 \}),$ \\ 

$(16, \{ 17, 23, 26, 28, 30, 31, 32, 33, 36, 39, 40, 41, 42, 44, 45, 46, 57, 58, 60, 62, 63, 66, 68, 70, 71, 72, 73, 75, 76, 79, 80 \}),$ \\ 
$(17, \{ 18, 19, 20, 21, 22, 23, 24, 25, 26, 29, 32, 34, 35, 37, 39, 40, 42, 43, 47, 49, 51, 53, 55, 56, 65, 66, 68, 69, 72, 74, 80 \}),$ \\ 
$(18, \{ 19, 20, 21, 22, 23, 24, 27, 28, 32, 33, 34, 35, 37, 38, 41, 42, 44, 46, 48, 49, 51, 52, 54, 56, 67, 68, 70, 71, 72, 73, 74, 75, 78 \}),$ \\ 
$(19, \{ 20, 21, 22, 23, 24, 25, 26, 27, 29, 31, 33, 36, 37, 39, 40, 42, 44, 50, 51, 52, 54, 55, 65, 66, 67, 69, 71, 73, 76, 77 \}),$ \\ 
$(20, \{ 21, 22, 23, 24, 27, 29, 34, 35, 36, 38, 40, 43, 44, 45, 48, 49, 50, 53, 54, 56, 66, 67, 69, 70, 72, 74, 75, 76, 78, 80 \}),$ \\ 
			
$(21, \{ 22, 23, 24, 25, 28, 29, 30, 33, 35, 36, 38, 39, 43, 46, 47, 49, 51, 52, 53, 55, 65, 68, 69, 70, 72, 76, 78 \}),$ \\ 
$(22, \{ 23, 24, 28, 31, 32, 33, 34, 37, 38, 39, 42, 44, 45, 46, 48, 50, 52, 53, 55, 56, 66, 67, 68, 71, 72, 74, 77, 78, 79 \}),$ \\ 
$(23, \{ 24, 25, 27, 29, 30, 31, 33, 35, 36, 37, 40, 46, 48, 50, 51, 54, 55, 56, 65, 67, 69, 70, 71, 73, 77, 80 \}),$ \\ 
$(24, \{ 25, 31, 34, 36, 38, 39, 40, 41, 44, 47, 48, 49, 50, 52, 53, 54, 65, 66, 68, 70, 71, 74, 76, 78, 79, 80 \}),$ \\ 
$(25, \{ 26, 27, 28, 29, 30, 31, 32, 33, 34, 37, 40, 42, 43, 45, 47, 48, 50, 51, 55, 57, 59, 61, 63, 64, 73, 74, 76, 77, 80 \}),$ \\ 

$(26, \{ 27, 28, 29, 30, 31, 32, 35, 36, 40, 41, 42, 43, 45, 46, 49, 50, 52, 54, 56, 57, 59, 60, 62, 64, 75, 76, 78, 79, 80 \}),$ \\ 
$(27, \{ 28, 29, 30, 31, 32, 33, 34, 35, 37, 39, 41, 44, 45, 47, 48, 50, 52, 58, 59, 60, 62, 63, 73, 74, 75, 77, 79 \}),$ \\ 
$(28, \{ 29, 30, 31, 32, 35, 37, 42, 43, 44, 46, 48, 51, 52, 53, 56, 57, 58, 61, 62, 64, 74, 75, 77, 78, 80 \}),$ \\ 
$(29, \{ 30, 31, 32, 33, 36, 37, 38, 41, 43, 44, 46, 47, 51, 54, 55, 57, 59, 60, 61, 63, 73, 76, 77, 78, 80 \}),$ \\ 
$(30, \{ 31, 32, 36, 39, 40, 41, 42, 45, 46, 47, 50, 52, 53, 54, 56, 58, 60, 61, 63, 64, 74, 75, 76, 79, 80 \}),$ \\ 
			
$(31, \{ 32, 33, 35, 37, 38, 39, 41, 43, 44, 45, 48, 54, 56, 58, 59, 62, 63, 64, 73, 75, 77, 78, 79 \}),$ \\
$(32, \{ 33, 39, 42, 44, 46, 47, 48, 49, 52, 55, 56, 57, 58, 60, 61, 62, 73, 74, 76, 78, 79 \}),$ \\ 
$(33, \{ 34, 35, 36, 37, 38, 39, 40, 41, 42, 45, 48, 50, 51, 53, 55, 56, 58, 59, 63, 65, 67, 69, 71, 72 \}),$ \\ 
$(34, \{ 35, 36, 37, 38, 39, 40, 43, 44, 48, 49, 50, 51, 53, 54, 57, 58, 60, 62, 64, 65, 67, 68, 70, 72 \}),$ \\ 
$(35, \{ 36, 37, 38, 39, 40, 41, 42, 43, 45, 47, 49, 52, 53, 55, 56, 58, 60, 66, 67, 68, 70, 71 \}),$ \\ 

$(36, \{ 37, 38, 39, 40, 43, 45, 50, 51, 52, 54, 56, 59, 60, 61, 64, 65, 66, 69, 70, 72 \}),$ \\
$(37, \{ 38, 39, 40, 41, 44, 45, 46, 49, 51, 52, 54, 55, 59, 62, 63, 65, 67, 68, 69, 71 \}),$ \\ 
$(38, \{ 39, 40, 44, 47, 48, 49, 50, 53, 54, 55, 58, 60, 61, 62, 64, 66, 68, 69, 71, 72 \}),$ \\
$(39, \{ 40, 41, 43, 45, 46, 47, 49, 51, 52, 53, 56, 62, 64, 66, 67, 70, 71, 72 \}),$ \\
$(40, \{ 41, 47, 50, 52, 54, 55, 56, 57, 60, 63, 64, 65, 66, 68, 69, 70 \}),$ \\
$(41, \{ 42, 43, 44, 45, 46, 47, 48, 49, 50, 53, 56, 58, 59, 61, 63, 64, 66, 67, 71, 73, 75, 77, 79, 80 \}),$ \\ 

$(42, \{ 43, 44, 45, 46, 47, 48, 51, 52, 56, 57, 58, 59, 61, 62, 65, 66, 68, 70, 72, 73, 75, 76, 78, 80 \}),$ \\

$(43, \{ 44, 45, 46, 47, 48, 49, 50, 51, 53, 55, 57, 60, 61, 63, 64, 66, 68, 74, 75, 76, 78, 79 \}),$ \\
$(44, \{ 45, 46, 47, 48, 51, 53, 58, 59, 60, 62, 64, 67, 68, 69, 72, 73, 74, 77, 78, 80 \}),$ \\
$(45, \{ 46, 47, 48, 49, 52, 53, 54, 57, 59, 60, 62, 63, 67, 70, 71, 73, 75, 76, 77, 79 \}),$ \\
$(46, \{ 47, 48, 52, 55, 56, 57, 58, 61, 62, 63, 66, 68, 69, 70, 72, 74, 76, 77, 79, 80 \}),$ \\ 
$(47, \{ 48, 49, 51, 53, 54, 55, 57, 59, 60, 61, 64, 70, 72, 74, 75, 78, 79, 80 \}),
(48, \{49,55,58,60,62, 63, 64, 65, 68, 71, 72, 73, 74, 76, 77, 78 \}),$ \\ 
$(49, \{ 50, 51, 52, 53, 54, 55, 56, 57, 58, 61, 64, 66, 67, 69, 71, 72, 74, 75, 79 \}),$ \\
$(50, \{ 51, 52, 53, 54, 55, 56, 59, 60, 64, 65, 66, 67, 69, 70, 73, 74, 76, 78, 80 \}),$ \\
			
$(51, \{ 52, 53, 54, 55, 56, 57, 58, 59, 61, 63, 65, 68, 69, 71, 72, 74, 76 \}),
(52, \{ 53, 54, 55, 56, 59, 61, 66, 67, 68, 70, 72, 75, 76, 77, 80 \}),$ \\ 

$(53, \{ 54, 55, 56, 57, 60, 61, 62, 65, 67, 68, 70, 71, 75, 78, 79 \}),
(54, \{ 55, 56, 60, 63, 64, 65, 66, 69, 70, 71, 74, 76, 77, 78, 80 \})$\\ 

$(55, \{ 56, 57, 59, 61, 62, 63, 65, 67, 68, 69, 72, 78, 80 \}),
(56, \{ 57, 63, 66, 68, 70, 71, 72, 73, 76, 79, 80 \}),$ \\
$(57, \{ 58, 59, 60, 61, 62, 63, 64, 65, 66, 69, 72, 74, 75, 77, 79, 80 \}),
(58, \{ 59, 60, 61, 62, 63, 64, 67, 68, 72, 73, 74, 75, 77, 78 \}),$ \\

$(59, \{ 60, 61, 62, 63, 64, 65, 66, 67, 69, 71, 73, 76, 77, 79, 80 \}),
(60, \{ 61, 62, 63, 64, 67, 69, 74, 75, 76, 78, 80 \}),$ \\ 
			
$(61, \{ 62, 63, 64, 65, 68, 69, 70, 73, 75, 76, 78, 79 \}),
(62, \{ 63, 64, 68, 71, 72, 73, 74, 77, 78, 79 \}),$ \\

$(63, \{ 64, 65, 67, 69, 70, 71, 73, 75, 76, 77, 80 \}),
(64, \{ 65, 71, 74, 76, 78, 79, 80 \}),
(65, \{ 66, 67, 68, 69, 70, 71, 72, 73, 74, 77, 80 \}),$ \\
$(66, \{ 67, 68, 69, 70, 71, 72, 75, 76, 80 \}),
(67, \{ 68, 69, 70, 71, 72, 73, 74, 75, 77, 79 \}),
(68, \{ 69, 70, 71, 72, 75, 77 \}),$ \\
$(69, \{ 70, 71, 72, 73, 76, 77, 78 \}),
(70, \{ 71, 72, 76, 79, 80 \}),
(71, \{ 72, 73, 75, 77, 78, 79 \}),
(72, \{ 73, 79 \}),$ \\ 

$(73, \{ 74, 75, 76, 77, 78, 79, 80 \}),
(74, \{ 75, 76, 77, 78, 79, 80 \}),
(75, \{ 76, 77, 78, 79, 80 \}),$ \\
$(76, \{ 77, 78, 79, 80 \}),
(77, \{ 78, 79, 80 \}),
(78, \{ 79, 80 \}),(79, \{ 80 \}).$ \\ 
\bottomrule 
\end{tabular}
}
\end{table}

\section{Two new qubit codes}\label{sec:qubit}

We presents two new codes found via a randomized search for bordered metacirculant graphs. Calculating the minimum distances of $Q_{81,1}$ and $Q_{94}$ in {\tt MAGMA} \cite{BOSMA1997} took $6$ and $27$ days. Secondary constructions often derive qubit codes that perform better than the current best-known. They are explained, for examples, in \cite[Theorem 6]{Calderbank1998} and \cite[Section 4]{Grassl20}. These propagation rules are routinely performed on all new submissions to the online table and improved codes, if found, are then put on record. We omit the derivation of possibly new codes by secondary constructions.

\begin{proposition}\label{prop81}
For $k \in \{1,2,3\}$, each $\overline{G}_{80,k}$, with $G_{80,k}$ given respectively by
\begin{align}\label{eq:G81}
& G(8, 10, 7, \{ 1, 4, 6, 9 \},\{ 0, 1, 2, 3, 6, 7, 8 \}, 
\{ 0, 2, 3, 4, 8, 9 \},\{ 0, 6 \},\{ 0, 1, 2, 4, 6, 8, 9 \}),\\
& G(8, 10, 3, \{ 4, 5, 6 \},
\{ 0, 1, 2, 4, 5, 7, 9 \},\{ 0, 1, 5, 8, 9 \}, 
\{ 0, 2, 3, 7 \},\{ 1, 2, 3, 4, 5, 6, 7, 8, 9 \}),\\
& G(10, 8, 5, \{ 2, 3, 5, 6 \},\{3\},
\{ 2, 4, 6, 7 \}, \{ 5, 6 \}, 
\{ 0, 1, 2, 3, 4, 6 \},\{ 0, 2, 5, 6, 7 \}),
\end{align}
generates a new $(81,2^{81}, 20)_4$ Type I additive self-dual code. The corresponding $\dsb{81, 0, 20}_2$ code $Q_{81,k}$ improves on the $\dsb{81, 0, 19}_2$ code in \cite{Grassl:codetables}. Table~\ref{graph:80} lists the edges of $G_{80,1}$. 
\end{proposition}

\begin{proposition}\label{prop93}
The bordered metacirculant graph $\overline{G}_{93}$ with
\begin{equation}\label{eq:G94}
G_{93}:= G(3, 31, 1, 
\{ 10, 12, 13, 15, 16, 18, 19, 21 \},
\{ 4, 6, 7, 9, 12, 14, 15, 18, 19, 21 \})
\end{equation}
generates a new $\left(94,2^{94}, 22\right)_4$ Type II additive self-dual code $C_{94}$. The corresponding $\dsb{94, 0, 22}_2$ code $Q_{94}$ improves the minimum distance of the $\dsb{94, 0, 20}$ code in \cite{Grassl:codetables} by $2$. The vertex set of $G_{93}$ are partitioned into $V_0,V_1,V_2$, each containing the respective $31$ vertices in $\{(k,t) : k \in \{0,1,2\}, 0 \leq t < 31\}$. Table~\ref{graph:93} presents the edge set of $G_{93}$ based on the partition. 
\end{proposition}

\begin{table}
\caption{The $1\,302$ Edges of $G_{93}$. Vertices $1$ to $31$ relabel those in $V_0$. Vertices $32$ to $63$ stand in place of those in $V_1$. Vertices $64$ to $93$ rename the ones in $V_2$. The tuple $(r,s)$ indicates that $i \in V_r$ and $j \in V_s$ for $r,s \in \{0,1,2\}$. The presentation allows us to observe if an edge connects two vertices within the same $V_r$ or two vertices $i \in V_r$ and $j \in V_s$ with $r \neq s$.}
\label{graph:93}
\setlength{\tabcolsep}{2pt}
\renewcommand{\arraystretch}{1.1}
\centering
\resizebox{\textwidth}{!}{%
\begin{tabular}{cc l}
\toprule
$\#$ & $(r,s)$ & $(i,\{j \in J\})$ \\
\midrule
$48$ & $(0,0)$ & $(1, \{ 14, 20, 23, 29, 31 \}), 
(2, \{ 15, 21, 24, 30, 31 \}), (3, \{ 13, 19, 22, 28 \}), 
(4, \{ 17, 23, 26 \}), (5, \{ 18, 24, 27 \}), (6, \{ 16, 22, 25, 31 \}), $\\

& & $(7, \{ 20, 26, 29 \}),(8, \{ 21, 27, 30 \}),
(9, \{ 19, 25, 28 \}), (10, \{ 23, 29 \}),  
(11, \{ 24, 30 \}), (12, \{ 22, 28, 31 \}), (13, \{ 26 \}),  
(14, \{ 27 \}), $\\

&& $ (15, \{ 25, 31 \}), (16, \{ 29 \}),
(17, \{ 30 \}), (18, \{ 28 \}), (21, \{ 31 \}).$ \\ 
			
$49$ & $(1,1)$ & $(32, \{ 45, 51, 54, 60, 61, 62 \}), 
(33, \{ 43, 49, 52, 58, 62 \}), (34, \{ 47, 53, 56, 62 \}), 
(35, \{ 48, 54, 57 \}), (36, \{ 46, 52, 55, 61 \}), $\\

&& $(37, \{ 50, 56, 59 \}),(38, \{ 51, 57, 60 \}),
(39, \{ 49, 55, 58 \}), 
(40, \{ 53, 59, 62 \}), (41, \{ 54, 60 \}), 
(42, \{ 52, 58, 61 \}), (43, \{ 56, 62 \}),$\\

&& $ (44, \{ 57 \}),(45, \{ 55, 61 \}), (46, \{ 59 \}), (47, \{ 60 \}), 
(48, \{ 58 \}), (49, \{ 62 \}), (51, \{ 61 \}).$ \\
			
$48$ & $(2,2)$ & $(63, \{ 73, 79, 82, 88, 92, 93 \}), 
(64, \{ 77, 83, 86, 92 \}), (65, \{ 78, 84, 87, 93 \}), 
(66, \{ 76, 82, 85, 91 \}), (67, \{ 80, 86, 89 \}),$ \\
 
&& $(68, \{ 81, 87, 90 \}),(69, \{ 79, 85, 88 \}), 
(70, \{ 83, 89, 92 \}), (71, \{ 84, 90, 93 \}), 
(72, \{ 82, 88, 91 \}), (73, \{ 86, 92 \}), (74, \{ 87, 93 \}), $\\
&& $(75, \{ 85, 91 \}), (76, \{ 89 \}), 
(77, \{ 90 \}), (78, \{ 88 \}), (79, \{ 92 \}), 
(80, \{ 93 \}), (81, \{ 91 \}).$\\ 
			
			$385$ & $(0,1)$ & 
			$(1, \{ 33, 37, 38, 39, 40, 42, 44, 46, 47, 49, 51, 54, 55, 56, 58, 59, 60 \}), (2, \{ 32, 37, 38, 39, 40, 41, 45, 47, 48, 49, 50, 52, 56, 57, 58, 59, 60 \}),$\\
			
			&& $(3, \{ 32, 33, 37, 38, 39, 41, 42, 43, 46, 48, 50, 51, 53, 55, 57, 58, 59, 60 \}), (4, \{ 32, 34, 36, 40, 41, 42, 43, 45, 47, 49, 50, 52, 54, 57, 58, 59, 61, 62\}),$ \\
			
			&& $(5, \{ 33, 34, 35, 40, 41, 42, 43, 44, 48, 50, 51, 52, 53, 55, 59, 60, 61, 62 \}), (6, \{ 35, 36, 40, 41, 42, 44, 45, 46, 49, 51, 53, 54, 56, 58, 60, 61, 62 \}),$ \\
			
			&& $(7, \{ 35, 37, 39, 43, 44, 45, 46, 48, 50, 52, 53, 55, 57, 60, 61, 62\}), (8, \{ 36, 37, 38, 43, 44, 45, 46, 47, 51, 53, 54, 55, 56, 58, 62 \}),$ \\ 
			
			&& $(9, \{ 34, 38, 39, 43, 44, 45, 47, 48, 49, 52, 54, 56, 57, 59, 61 \}),  
			(10, \{ 32, 38, 40, 42, 46, 47, 48, 49, 51, 53, 55, 56, 58, 60 \}),$ \\ 
			
			&& $(11, \{ 33, 39, 40, 41, 46, 47, 48, 49, 50, 54, 56, 57, 58, 59, 61 \}), 
			(12, \{ 37, 41, 42, 46, 47, 48, 50, 51, 52, 55, 57, 59, 60, 62 \}),$ \\ 
			
			&& $(13, \{ 32, 35, 41, 43, 45, 49, 50, 51, 52, 54, 56, 58, 59, 61 \}), 
			(14, \{ 33, 36, 42, 43, 44, 49, 50, 51, 52, 53, 57, 59, 60, 61, 62 \}),$ \\ 
			
			&& $(15, \{ 34, 40, 44, 45, 49, 50, 51, 53, 54, 55, 58, 60, 62 \}), 
			(16, \{ 35, 38, 44, 46, 48, 52, 53, 54, 55, 57, 59, 61, 62 \}),$ \\ 
			
			&& $(17, \{ 36, 39, 45, 46, 47, 52, 53, 54, 55, 56, 60, 62 \}), 
			(18, \{ 34, 37, 43, 47, 48, 52, 53, 54, 56, 57, 58, 61 \}),$ \\
			
			&& $(19, \{ 32, 38, 41, 47, 49, 51, 55, 56, 57, 58, 60, 62 \}), 
			(20, \{ 33, 39, 42, 48, 49, 50, 55, 56, 57, 58, 59 \}), 
			(21, \{ 37, 40, 46, 50, 51, 55, 56, 57, 59, 60, 61 \}),$ \\
			
			&& $(22, \{ 35, 41, 44, 50, 52, 54, 58, 59, 60, 61 \}), 
			(23, \{ 36, 42, 45, 51, 52, 53, 58, 59, 60, 61, 62 \}), 
			(24, \{ 34, 40, 43, 49, 53, 54, 58, 59, 60, 62 \}),$ \\
			
			&& $(25, \{ 38, 44, 47, 53, 55, 57, 61, 62 \}), 
			(26, \{ 39, 45, 48, 54, 55, 56, 61, 62 \}), 
			(27, \{ 37, 43, 46, 52, 56, 57, 61, 62 \}),$ \\
			
			&& $(28, \{ 41, 47, 50, 56, 58, 60 \}), 
			(29, \{ 42, 48, 51, 57, 58, 59 \}), 
			(30, \{ 40, 46, 49, 55, 59, 60 \}), 
			(31, \{ 44, 50, 53, 59, 61 \}).$ \\ 
			
			$387$ & $(0,2)$ & 
			$(1, \{ 64, 65, 69, 75, 78, 84 \}), (2, \{ 65, 66, 67, 73, 76, 82 \}), 
			(3, \{ 64, 66, 68, 74, 77, 83 \}), (4, \{ 63, 67, 68, 72, 78, 81, 87 \}),$ \\
			
			&& $(5, \{ 63, 68, 69, 70, 76, 79, 85 \}), 
			(6, \{ 63, 67, 69, 71, 77, 80, 86 \}), 
			(7, \{ 64, 65, 66, 70, 71, 75, 81, 84, 90 \}), 
			(8, \{ 63, 64, 65, 66, 71, 72, 73, 79, 82, 88 \}),$ \\
			
			&& 
			$(9, \{ 63, 64, 65, 66, 70, 72, 74, 80, 83, 89 \}), 
			(10, \{ 63, 64, 65, 67, 68, 69, 73, 74, 78, 84, 87, 93 \}), 
			(11, \{ 65, 66, 67, 68, 69, 74, 75, 76, 82, 85, 91 \}),$ \\ 
			
			&& 
			$(12, \{ 64, 66, 67, 68, 69, 73, 75, 77, 83, 86, 92 \}), 
			(13, \{ 63, 66, 67, 68, 70, 71, 72, 76, 77, 81, 87, 90 \}),$ \\ 
			
			&& $(14, \{ 64, 68, 69, 70, 71, 72, 77, 78, 79, 85, 88 \}), 
			(15, \{ 63, 65, 67, 69, 70, 71, 72, 76, 78, 80, 86, 89 \}),$ \\ 
			
			&& $(16, \{ 64, 66, 69, 70, 71, 73, 74, 75, 79, 80, 84, 90, 93 \}), 
			(17, \{ 63, 64, 65, 67, 71, 72, 73, 74, 75, 80, 81, 82, 88, 91 \}),$ \\ 
			
			&& $(18, \{ 63, 65, 66, 68, 70, 72, 73, 74, 75, 79, 81, 83, 89, 92 \}), 
			(19, \{ 64, 65, 67, 69, 72, 73, 74, 76, 77, 78, 82, 83, 87, 93 \}),$ \\ 
			
			&& $(20, \{ 63, 65, 66, 67, 68, 70, 74, 75, 76, 77, 78, 83, 84, 85, 91 \}), 
			(21, \{ 64, 66, 68, 69, 71, 73, 75, 76, 77, 78, 82, 84, 86, 92 \}),$ \\ 
			
			&& $(22, \{ 63, 65, 67, 68, 70, 72, 75, 76, 77, 79, 80, 81, 85, 86, 90 \}), 
			(23, \{ 66, 68, 69, 70, 71, 73, 77, 78, 79, 80, 81, 86, 87, 88 \}),$ \\ 
			
			&& $(24, \{ 63, 64, 67, 69, 71, 72, 74, 76, 78, 79, 80, 81, 85, 87, 89 \}), 
			(25, \{ 63, 64, 66, 68, 70, 71, 73, 75, 78, 79, 80, 82, 83, 84, 88, 89, 93 \}),$ \\ 
			
			&& $(26, \{ 63, 64, 65, 69, 71, 72, 73, 74, 76, 80, 81, 82, 83, 84, 89, 90, 91 \}), 
			(27, \{ 63, 65, 66, 67, 70, 72, 74, 75, 77, 79, 81, 82, 83, 84, 88, 90, 92 \}),$ \\ 
			
			&& $(28, \{ 64, 65, 66, 67, 69, 71, 73, 74, 76, 78, 81, 82, 83, 85, 86, 87, 91, 92 \}), 
			(29, \{ 64, 65, 66, 67, 68, 72, 74, 75, 76, 77, 79, 83, 84, 85, 86, 87, 92, 93 \}),$ \\ 
			&& $(30, \{ 64, 65, 66, 68, 69, 70, 73, 75, 77, 78, 80, 82, 84, 85, 86, 87, 91, 93 \}), 
			(31, \{ 63, 67, 68, 69, 70, 72, 74, 76, 77, 79, 81, 84, 85, 86, 88, 89, 90 \}).$ \\ 
			
			$385$ & $(1,2)$ &
			$(32, \{ 67, 68, 69, 70, 71, 75, 77, 78, 79, 80, 82, 86, 87, 88, 89, 90 \}), 
			(33, \{ 63, 67, 68, 69, 71, 72, 73, 76, 78, 80, 81, 83, 85, 87, 88, 89, 90 \}),$ \\ 
			
			&& $(34, \{ 64, 66, 70, 71, 72, 73, 75, 77, 79, 80, 82, 84, 87, 88, 89, 91, 92, 93 \}), 
			(35, \{ 63, 64, 65, 70, 71, 72, 73, 74, 78, 80, 81, 82, 83, 85, 89, 90, 91, 92, 93 \}),$ \\
			
			&& $(36, \{ 65, 66, 70, 71, 72, 74, 75, 76, 79, 81, 83, 84, 86, 88, 90, 91, 92, 93 \}), 
			(37, \{ 65, 67, 69, 73, 74, 75, 76, 78, 80, 82, 83, 85, 87, 90, 91, 92 \}),$ \\ 
			
			&& $(38, \{ 66, 67, 68, 73, 74, 75, 76, 77, 81, 83, 84, 85, 86, 88, 92, 93 \}), (39, \{ 64, 68, 69, 73, 74, 75, 77, 78, 79, 82, 84, 86, 87, 89, 91, 93 \}),$ \\ 
			
			&& $(40, \{ 68, 70, 72, 76, 77, 78, 79, 81, 83, 85, 86, 88, 90, 93 \}), 
			(41, \{ 63, 69, 70, 71, 76, 77, 78, 79, 80, 84, 86, 87, 88, 89, 91 \}),$ \\ 
						
			&& $(42, \{ 67, 71, 72, 76, 77, 78, 80, 81, 82, 85, 87, 89, 90, 92 \}), 
			(43, \{ 65, 71, 73, 75, 79, 80, 81, 82, 84, 86, 88, 89, 91, 93 \}),$ \\ 
			
			&& $(44, \{ 63, 66, 72, 73, 74, 79, 80, 81, 82, 83, 87, 89, 90, 91, 92 \}), 
			(45, \{ 64, 70, 74, 75, 79, 80, 81, 83, 84, 85, 88, 90, 92, 93 \}),$ \\ 
			
			&& $(46, \{ 65, 68, 74, 76, 78, 82, 83, 84, 85, 87, 89, 91, 92 \}), 
			(47, \{ 66, 69, 75, 76, 77, 82, 83, 84, 85, 86, 90, 92, 93 \}),$ \\ 
			
			&& $(48, \{ 64, 67, 73, 77, 78, 82, 83, 84, 86, 87, 88, 91, 93 \}), 
			(49, \{ 68, 71, 77, 79, 81, 85, 86, 87, 88, 90, 92 \}),$ \\ 
			
			&& $(50, \{ 63, 69, 72, 78, 79, 80, 85, 86, 87, 88, 89, 93 \}), 
			(51, \{ 67, 70, 76, 80, 81, 85, 86, 87, 89, 90, 91 \}), 
			(52, \{ 65, 71, 74, 80, 82, 84, 88, 89, 90, 91, 93 \}),$ \\ 
			
			&& $(53, \{ 66, 72, 75, 81, 82, 83, 88, 89, 90, 91, 92 \}), 
			(54, \{ 64, 70, 73, 79, 83, 84, 88, 89, 90, 92, 93 \}), 
			(55, \{ 68, 74, 77, 83, 85, 87, 91, 92, 93 \}),$ \\ 
			
			&& $(56, \{ 69, 75, 78, 84, 85, 86, 91, 92, 93 \}), 
			(57, \{ 67, 73, 76, 82, 86, 87, 91, 92, 93 \}), 
			(58, \{ 71, 77, 80, 86, 88, 90 \}),$ \\ 
			
			&& $(59, \{ 72, 78, 81, 87, 88, 89 \}), 
			(60, \{ 70, 76, 79, 85, 89, 90 \}), 
			(61, \{ 74, 80, 83, 89, 91, 93 \}), 
			(62, \{ 75, 81, 84, 90, 91, 92 \}).$ \\ 
			\bottomrule 
		\end{tabular} 
	}
\end{table}

Table~\ref{table:property} gives the properties of the graphs. All have diameter $2$ and girth $3$, except for $G_{36,1}$ whose diameter is $3$. Listed are the minimum distance $d_{\rm min}(G)$ of $C:=C(G)$, the valency $\nu(G)$, the maximum clique size $\gamma(G)$, and the size $|{\rm Aut}(G)|$ of the automorphism group.

\begin{table}
\caption{Properties of the Graphs}
\label{table:property}
\setlength{\tabcolsep}{5pt}
\centering
\begin{tabular}{l c c c c | l c c c c }
\toprule
Graph & $d_{\rm min}(G)$ & $\nu(G)$ & $\gamma(G)$ & $|{\rm Aut}(G)|$ &
Graph & $d_{\rm min}(G)$ & $\nu(G)$ & $\gamma(G)$ & $|{\rm Aut}(G)|$ \\
\midrule
$G_{28}$ & $10$ & $11$ & $4$ & $28$ &
$G_{80,1}$ & $20$ & $41$ & $8$ & $80$ \\

$G_{36,1}$ & $11$ & $11$ & $4$ & $36$ &
$G_{80,2}$ & $20$ & $44$ & $7$ & $80$ \\

$G_{36,2}$ & $11$ & $19$ & $5$ & $36$ &
$G_{80,3}$ & $20$ & $35$ & $9$ & $80$ \\

&&&&&	
$G_{93}$ & $22$ & $28$ & $4$ & $186$ \\
\bottomrule 
\end{tabular} 
\end{table}

\section*{Acknowledgements}
Nanyang Technological University Grant Number 04INS000047C230GRT01 supports the research carried out by M. F. Ezerman.

\bibliographystyle{plain}

\end{document}